\documentclass[letterpaper, 10 pt, conference]{ieeeconf}  

\IEEEoverridecommandlockouts                              

\usepackage{graphicx}          
\usepackage[utf8]{inputenc}
\usepackage{textcomp}

\usepackage{graphicx}
\usepackage{amsmath}

\usepackage{siunitx}
\usepackage{longtable,tabularx}
\usepackage[caption=false,font=footnotesize]{subfig}
\usepackage{algorithm}
\usepackage{algpseudocode}
\usepackage{xcolor}
\usepackage{comment}
\usepackage{upgreek}
\usepackage{kotex}
\usepackage{graphicx}
\usepackage{booktabs}
\usepackage{hyperref}
\usepackage{setspace}
\usepackage{mathtools}
\usepackage{amsfonts}
\usepackage{amssymb}
\usepackage{bigints}
\usepackage{xspace}
\usepackage{cite}
\let\labelindent\relax
\usepackage{enumitem}
\usepackage{tikz}
\usetikzlibrary{arrows.meta,positioning,fit,calc}

\graphicspath{{images/}}

\newtheorem{theorem}{Theorem}       
\newtheorem{proposition}[theorem]{Proposition}

\newtheorem{definition}[theorem]{Definition}

\newtheorem{remark}{Remark}
\newtheorem{statement}{Statement}

\title{\LARGE \bf
Learning Safety--Compatible Observers for Unknown Systems
}

\author{Juho Bae$^{*}$, Daegyeong Roh$^{*}$, and Han-Lim Choi$^{\dagger}$
\thanks{$^*$These authors equally contributed to the paper.}
\thanks{The authors are with the Department of Aerospace Engineering, KAIST, 34141, Daejon, South Korea,
        {\tt\small \{johnbae1901, dgroh2000, hanlimc\}@kaist.ac.kr}}%
\thanks{$^{\dagger}$Corresponding author.}%
}

\begin{document}

\maketitle
\thispagestyle{empty}
\pagestyle{empty}

\begin{abstract}
    This paper presents a data--driven approach for jointly learning a robust full--state observer and its robustness certificate for systems with unknown dynamics. Leveraging incremental input--to--state stability ($\delta$ISS) notions, we jointly learn a $\delta$ISS Lyapunov function that serves as the robustness certificate and prove practical convergence of the estimation error under standard fidelity assumptions on the learned models. This renders the observer \emph{safety--compatible}: they can be consumed by certificate--based safe controllers so that, when the controller tolerates bounded estimation error, the controller’s certificate remains valid under output feedback. We further extend the approach to interconnected systems via the small--gain theorem, yielding a distributed observer design framework. We validate the approach on a variety of nonlinear systems.
\end{abstract}
%
\section{Introduction}
%

%
%

Recent advances in learning--based control deliver strong performance on difficult control problems, but often at the cost of reduced interpretability and the absence of formal safety or stability guarantees. In response, new techniques in the field of \emph{certificate learning} have emerged to learn certificates alongside the controllers~\cite{dawson2023}. These include \emph{Lyapunov functions}~\cite{min2023} for stability, \emph{barrier functions}~\cite{peruffo2021automated} for positive invariance, and \emph{contraction metric}~\cite{tsukamoto2021} for incremental stability. Such methods have been extended to settings with partially or completely unknown nonlinear dynamics~\cite{jagtap2020control,min2023} and to high--dimensional networked systems~\cite{zhang2023compositional}.

However, safety or stability guarantees of most certificate--based controllers assume access to the true state, creating a gap between simulation and real--world application where the state must be inferred from measurements. One line of work addresses this gap by synthesizing measurement--robust barrier functions that allow bounded estimation error~\cite{dean2021guaranteeing}. Nevertheless, beyond controller--side robustness, reliable output--feedback still requires \emph{observers} that provide state estimates with quantified error bounds compatible with the controller’s certificate.
%
%

Constructing observers for nonlinear systems\textemdash{}particularly in settings analogous to those studied for controllers (e.g., partially or completely unknown dynamics)\textemdash{}remains an open problem with no universal constructive approach~\cite{bernard2022observer,peralez2024deep}. Ref.~\cite{agrawal2022safe} considers synthesizing an observer--controller interconnection for known control--affine systems using control barrier functions. Ref.~\cite{tsukamoto2020neural} develops a model--based approach that parameterizes a contraction metric with neural networks to obtain state--estimation schemes for autonomous systems. Many recent works are grounded in Kazantzis--Kravaris--Luenberger (KKL) theory~\cite{kazantzis1998nonlinear}, which maps the nonlinear system into an auxiliary higher--dimensional stable linear system and learns (or otherwise synthesizes) the coordinate transform. KKL--based observers have been studied via supervised learning~\cite{ramos2020numerical}, physics--informed neural networks~\cite{niazi2022learning}, neural ODEs~\cite{miao2023learning}, and switching among multiple Deep--KKL observers~\cite{peralez2024deep}. However, the above lines of work often (i) require full or partial model knowledge (e.g., control--affine structure), (ii) assume autonomous dynamics without control inputs, or (iii) do not provide explicit disturbance--to--estimation--error gain bounds. Moreover, apart from contraction metric--based approaches~\cite{tsukamoto2020neural}, there remains a need for data--driven certificates that can prove learned observers to be \emph{safety--compatible} with certificate--based controllers in safety--critical applications.
%
%

To this end, we adopt a certificate--learning viewpoint and propose a data--driven method that jointly learns (i) the dynamics model, (ii) a robust full--state observer, and (iii) an incremental input--to--state stability ($\delta$ISS) Lyapunov certificate. First, building on $\delta$ISS Lyapunov theory~\cite{angeli2002}, we introduce a hard--constrained neural network architecture~\cite{min2024} that enforces the $\delta$ISS decrease condition by construction, yielding explicit $\mathcal{L}_\infty$--gain bounds from input/output disturbances to the estimation error. Under standard fidelity assumptions on the learned models, we prove practical convergence of the estimate to the true state. Second, we extend the approach to a two--subsystem interconnection via small--gain results~\cite{jiang1994,angeli2002}, yielding a distributed design in which each local observer is synthesized without access to the other subsystem’s dynamics. While our analysis focuses on two subsystems, the same template suggests a path toward $N$--subsystem networks. A full treatment is left for future work. 
%
\section{Preliminaries} \label{sec:preliminaries}
\subsection{Incremental input--to--state stability}
Throughout this paper, we consider a nonlinear system 
\begin{equation} \label{eq:system}
    \dot{x}(t) = f(x(t), u(t)), \quad y(t) = h(x(t))
\end{equation}
with state $x \in \mathbb{R}^{n_{x}}$, output $y \in \mathbb{R}^{n_{y}}$, and input $u \in \mathcal{U} \subseteq \mathbb{R}^{n_{u}}$. $f: \mathbb{R}^{n_x} \times \mathcal{U} \mapsto \mathbb{R}^{n_x}$ and $h: \mathbb{R}^{n_x} \mapsto \mathbb{R}^{n_y}$ are assumed to be sufficiently smooth. Inputs are measurable and essentially bounded signals $u: [0, \infty) \mapsto \mathcal{U}$. We use $|\cdot|$ for the Euclidean norm of vectors, $\|\cdot\|_{\infty}$ for the supremum norm, and $\|v\|_{[0,t]}:=\sup_{s\in[0,t]}\|v(s)\|$. We denote the state solution of~\eqref{eq:system} with initial value ${\xi} \in \mathbb{R}^{n_{x}}$ and input signal $u(t)$ as $x(t, \xi, u)$. The comparison--function notations follow the standard definition~\cite{angeli2002}.
\begin{definition} \label{def:deltaiss}
    System~\eqref{eq:system} is \textit{incrementally input--to--state stable} ($\delta$ISS) with respect to input $u$ if there exist $\beta \in \mathcal{KL}$ and $\gamma \in \mathcal{K}_{\infty}$ such that for all $t \geq 0$, $\xi_1$ and $\xi_2 \in \mathbb{R}^{n_{x}}$, and input signals $u_1$ and $u_2$, the following inequality holds~\cite{angeli2002}:
    \begin{align}
        \left| x(t, \xi_1, u_1) - x(t, \xi_2, u_2) \right|  \leq \beta & \left( |\xi_1 - \xi_2|, t \right) \\
        & + \gamma\left( \|u_1 - u_2\|_{\infty} \right) \tag*{$\square$}
    \end{align}
\end{definition}
When the control restraint set $\mathcal{U}$ is compact, a necessary and sufficient condition for $\delta$ISS of system~\eqref{eq:system} is existence of the following \textit{$\delta$ISS Lyapunov function}, while it remains as a sufficient condition for non-compact $\mathcal{U}$~{\cite[Thm.~2]{angeli2002}}.
\begin{definition} \label{def:isslyapunov}
    A smooth map $V(x_1, x_2): \mathbb{R}^{n_x} \times \mathbb{R}^{n_x} \mapsto \mathbb{R}_{\geq 0}$ is called a \textit{$\delta$ISS Lyapunov function} if $\alpha_1(|x_1 - x_2|) \leq V(x_1, x_2) \leq \alpha_2(|x_1 - x_2|)$ for some $\alpha_1, \alpha_2 \in \mathcal{K}_{\infty}$, and there exists $\kappa \in \mathcal{K}_{\infty}$ and $\rho \in \mathcal{K}$ such that
    \begin{align}
        \kappa(|x_1 - x_2|) \geq |u_1 - u_2| \Rightarrow & \frac{\partial V}{\partial x_1} f(x_1, u_1) + \frac{\partial V}{\partial x_2} f(x_2, u_2) \notag \\
        & < -\rho(|x_1 - x_2|)
    \end{align}
    for all $u_1, u_2 \in \mathcal{U}$ and $x_1, x_2 \in \mathbb{R}^{n_x}$. \hfill $\square$
\end{definition}

We next recall a small--gain result for interconnected systems~\cite{angeli2002,jiang1994}.
\begin{theorem} [{\cite[Prop.~4.8]{angeli2002}}] \label{thm:interconnected}
    Consider the following interconnection of two subsystems
    \begin{subequations} \label{eq:interconnected_system}
    \begin{align}
        \dot{x}_a = f_a(x_a, x_b, u_a), \label{eq:interconnected_system_a}\\
        \dot{x}_b = f_b(x_a, x_b, u_b). \label{eq:interconnected_system_b}
    \end{align}
    \end{subequations}
    Here the input to the first subsystem is seen as $(x_b,u_a)$, and $(x_a,u_b)$ to the second. Denote by $x_{a}\left(t,\xi_{a},x_{b},u_a\right)$ (resp. $x_{b}\left(t,\xi_{b},x_{a},u_b\right)$) the solution of the first (resp. second) subsystem with initial state $\xi_a$ (resp. $\xi_b$) and input $(x_b, u_a)$ (resp. $(x_a, u_b)$). Suppose each subsystem is $\delta$ISS with respect to their respective inputs, i.e., there exist $\gamma_{x_b},\gamma_{x_a},\gamma_{u_a},\gamma_{u_b} \in \mathcal{K}_{\infty}$ and $\beta_{a},\beta_{b} \in \mathcal{KL}$ such that for initial states $\xi_{a}^1$, $\xi_{a}^2$, $\xi_{b}^1$, $\xi_{b}^2$, and inputs $x_a^1$, $x_a^2$, $x_b^1$, $x_b^2$, $u^1$, $u^2$, the following hold for all $t \geq 0$: 
    \begin{align}
        & \left| x_{a}\left(t,\xi_{a}^1,x_{b}^1,u^1\right) - x_{a}\left(t,\xi_{a}^2,x_{b}^2,u^2\right) \right| \leq \beta_{a}\left(\left|\xi_{a}^1-\xi_{a}^2\right|,t\right) \notag \\
        & \quad + \gamma_{x_b}\left(\left\|x_{b}^1-x_{b}^2\right\|_{\infty}\right) 
        + \gamma_{u_a}(\left\|u_a^1-u_a^2\right\|_{\infty}), \label{eq:subsystem_a_delta_iss} \\
        & \left| x_{b}\left(t,\xi_{b}^1,x_{a}^1,u^1\right) - x_{b}\left(t,\xi_{b}^2,x_{a}^2,u^2\right) \right| \leq \beta_{b}\left(\left|\xi_{b}^1-\xi_{b}^2\right|,t\right) \notag \\
        & \quad + \gamma_{x_a}\left(\left\|x_{a}^1-x_{a}^2\right\|_{\infty}\right) 
        + \gamma_{u_b}(\left\|u_b^1-u_b^2\right\|_{\infty}). \label{eq:subsystem_b_delta_iss}
    \end{align}
    If there exists $\rho \in \mathcal{K}_{\infty}$ satisfying the small-gain condition:
    \begin{equation}
    (\gamma_{x_b}+\rho)\circ(\gamma_{x_a}+\rho)(\eta) \leq \eta, 
    \quad \forall \eta \geq 0,
    \end{equation}
    then the system~\eqref{eq:interconnected_system} is $\delta$ISS with respect to $u$. \hfill $\square$
\end{theorem}
In Theorem~\ref{thm:interconnected}, $\gamma_{x_b}$ and $\gamma_{x_a}$ are the gain functions for the interconnection inputs $x_b$ and $x_a$, respectively, whereas $\gamma_{u_a}$ and $\gamma_{u_b}$ denote the gains associated with the inputs $u_a$ and $u_b$ for respective subsystems. Moreover, Theorem~2.1 in~\cite{jiang1994} leads to an extension of Theorem~\ref{thm:interconnected} to \emph{practically $\delta$ISS} systems: if each subsystem is practically $\delta$ISS\textemdash{}i.e., \eqref{eq:subsystem_a_delta_iss} and \eqref{eq:subsystem_b_delta_iss} hold with additional additive residuals $d_a,d_b \ge 0$ on the right--hand side\textemdash{}with small--gain condition~\eqref{eq:small_gain}, then the interconnection \eqref{eq:interconnected_system} is also practically $\delta$ISS.
%
\subsection{Robust observers} \label{subsec:robust_observers}
\begin{definition}[Robust full--state observer, {\cite{sontag1997}}]
    A system 
    \begin{equation}\label{eq:observer_system_disturbance}
    \dot z \;=\; f\bigl(z,\,u+d_u\bigr) \;+\; L\bigl(z,\,u+d_u,\,y - h(z) +d_y\bigr)
    \end{equation}
    is called a \textit{robust full--state observer} (with respect to input disturbance $d_u$ and output measurement disturbance $d_y$) with $z(t)\in\mathbb{R}^{n_x}$ as the estimate of $x(t)$ if the followings hold:
    \begin{description}[leftmargin=*,style=nextline]
    \item[(i)] for all $t\ge 0$,
    \begin{align}
    & \bigl|x(t,\xi,u)-z(t,\zeta,u+d_u,h(x)+d_y)\bigr| \\
    &\le \beta(|\xi-\zeta|,t) + \gamma_1(\|d_u\|_{[0,t]}) + \gamma_2(\|d_y\|_{[0,t]}), \notag
    \end{align}
    \item[(ii)] $L(z,u,0)=0$ for all $z,u$. \hfill $\square$
    \end{description}
\end{definition}

The second condition is merely a necessary consistency condition: in the special case of $\xi=\zeta$ and $d_u=d_y\equiv 0$, the observer dynamics must satisfy $x(t)\equiv z(t)$ for all $t$. The following sufficient condition for existence of such observer is stated and proved in~\cite{angeli2002}.
\begin{theorem} [{\cite[Prop.~6.1]{angeli2002}}]\label{thm:observer}
    If the system
    \begin{equation} \label{eq:observer_system}
    \dot z \;=\; f\bigl(z,\,u\bigr) \;+\; L\bigl(z,\,u,\,y - h(z) \bigr)
    \end{equation}
    is $\delta$ISS with respect to $(u, y)$ and $L(z, u, 0) \equiv 0$, then~\eqref{eq:observer_system} is a robust full--state observer for~\eqref{eq:system}. \hfill $\square$
\end{theorem}

Theorem~\ref{thm:observer} implies that, to construct a robust observer against disturbances on $(u,y)$, it suffices to design an \emph{output injection} $L(z,u,y-h(z))$ that renders \eqref{eq:observer_system} $\delta$ISS.

Now consider again the interconnected system \eqref{eq:interconnected_system}. Assume each subsystem has outputs $y_a=h_a(x_a)$ and $y_b=h_b(x_b)$, and inputs $(x_b, u_a)$ and $(x_a, u_b)$, respectively. The overall plant takes $u \coloneqq (u_1, u_2)$ as input and produces the output $\bigl(h_a(x_a),\,h_b(x_b)\bigr)$. We construct a \emph{distributed observer} for \eqref{eq:interconnected_system} by combining Theorem~\ref{thm:interconnected} with Theorem~\ref{thm:observer}, where \emph{distributed} here means that each subsystem runs its own local observer and uses the other’s estimate.
\paragraph*{\textbf{Step A} (local design with exogenous interconnection signals)}
For an exogenous signal $w_b(\cdot)$, define the local observer of the form in Theorem~\ref{thm:observer},
\begin{equation}\label{eq:local_observer_a}
    \dot z_a \;=\; f_a\bigl(z_a,\,w_b,\,u_a\bigr) \;+\;
    L_a\!\bigl(z_a,\,w_b,\,u_a,\,y_a-h_a(z_a)\bigr),
\end{equation}
and suppose it is $\delta$ISS with respect to $(u_a,y_a,w_b)$. Define $z_b$ analogously using an exogenous $w_a(\cdot)$:
\begin{equation}\label{eq:local_observer_b}
\dot z_b \;=\; f_b\bigl(w_a,\,z_b,\,u_b\bigr) \;+\;
L_b\!\bigl(w_a,\,z_b,\,u_b,\,y_b-h_b(z_b)\bigr),
\end{equation}
and suppose it is $\delta$ISS with respect to $(u,y_b,w_a)$.
\paragraph*{\textbf{Step B} (interconnection)}
In the coupled implementation we take $w_b(\cdot)=z_b(\cdot)$ and $w_a(\cdot)=z_a(\cdot)$, yielding
\begin{subequations} \label{eq:distributed_observer}
\begin{align}
\dot{z}_a &= f_a(z_a, z_b, u_a) + L_a\bigl(z_a, z_b, u_a, y_a - h_a(z_a)\bigr), \label{eq:distributed_observer_a}\\
\dot{z}_b &= f_b(z_a, z_b, u_b) + L_b\bigl(z_a, z_b, u_b, y_b - h_b(z_b)\bigr). \label{eq:distributed_observer_b}
\end{align}
\end{subequations}
If, hypothetically, $z_b(\cdot)\equiv x_b(\cdot)$ and there are no disturbances acting on $u_a$ and $y_a$, then by Theorem~\ref{thm:observer} the local observer \eqref{eq:distributed_observer_a} satisfies $z_a(t)\to x_a(t)$ as $t\to\infty$. The analogous statement holds for \eqref{eq:distributed_observer_b}. In practice, perfect knowledge of the other subsystem’s state is not available, but each local observer only receives the other’s \emph{estimate}. However, if one can construct output injections $L_a$ and $L_b$ so that the local observers \eqref{eq:distributed_observer_a}--\eqref{eq:distributed_observer_b} satisfy the small--gain condition in Theorem~\ref{thm:interconnected}, then the interconnected observer \eqref{eq:distributed_observer} is $\delta$ISS with respect to $(u_a, u_b, y_a, y_b)$. Consequently, applying Theorem~\ref{thm:observer} to the interconnected plant~\eqref{eq:interconnected_system} and observer~\eqref{eq:distributed_observer} shows that \eqref{eq:distributed_observer} is a robust full--state observer for \eqref{eq:interconnected_system}, with robustness against input disturbances on $(u_a, u_b)$ and output disturbances on $(y_a,y_b)$. It is noteworthy that construction of $L_a$ and $L_b$ relies only on local models together with bounds on the interconnection gains. Explicit knowledge of the other subsystem’s dynamics is not required for the local design.
\subsection{Hard--constrained neural networks}
Enforcing \emph{input-dependent output constraints}\textemdash{}for each input $x$, the output must satisfy $f(x)\in\mathcal{C}(x)$\textemdash{}is crucial in safety--critical applications. Penalizing violations does not guarantee constraint satisfaction, which motivated architectures that enforce constraints \emph{by construction} (e.g.,~\cite{donti2021}); however, systematic guarantees on representational capacity\textemdash{}most notably, universal approximation under hard constraints\textemdash{}have been limited.

HardNet~\cite{min2024} addresses this by appending a differentiable, closed--form enforcement (projection) layer to the output of an unconstrained network, so constraints are satisfied \emph{by construction} while \emph{retaining universal approximation}. While HardNet supports more general convex constraints, for brevity we only present the affine case here.

Suppose the target function $f: D\subset\mathbb{R}^{n_x}\to\mathbb{R}^{n_{\text{out}}}$ must satisfy the affine constraints
\begin{equation}\label{eq:affine_constraints}
  b_1(x)\ \le\ A(x)\,f(x)\ \le\ b_2(x),\qquad \forall x\in\mathcal{X},
\end{equation}
where $A(x)\in\mathbb{R}^{n_c\times n_{\text{out}}}$ and $b_1(x),b_2(x)\in\mathbb{R}^{n_c}$ are continuous, so that there are $2n_c$ scalar inequalities. The domain $D$ is assumed compact. Assume the constraints are feasible for all $x \in D$ and that $A(x)$ has full row rank (hence $n_c\le n_{\text{out}}$). Then HardNet enforces~\eqref{eq:affine_constraints} by applying a \emph{projection}
\begin{align} \label{eq:hardnet_aff}
  \mathcal{P}(f_{\theta})(x) \coloneqq f_{\theta}(x)
  &+ A(x)^{\dagger}\!\left[\operatorname{ReLU}\!\bigl(b_1(x)-A(x)f_{\theta}(x)\bigr)\right. \notag\\
  &\left. - \operatorname{ReLU}\!\bigl(A(x)f_{\theta}(x)-b_2(x)\bigr)\right],
\end{align}
for all \(x \in D\), where $A^{\dagger}\equiv A^\top(AA^\top)^{-1}$ denotes the right pseudoinverse. This adjusts the violated components by moving the output in directions parallel to the boundaries of already--satisfied constraints, while leaving already--satisfied inequalities unchanged (cf.\ Proposition 4.4 in~\cite{min2024}). Moreover, if a base class of neural networks $\mathcal{F}_{\text{NN}}$ universally approximates $C(D,\mathbb{R}^{n_{\text{out}}})$ (or $L^p$), then the class $\mathcal{F}_{\text{HardNet}} \coloneqq \bigl\{\mathcal{P}(f_\theta)\,:\,f_\theta\in\mathcal{F}_{\text{NN}}\bigr\}$ universally approximates the subset of continuous (respectively $L^p$) functions that satisfy~\eqref{eq:affine_constraints} (Theorem~4.6 in~\cite{min2024}).

\cite{kolter2019} applied the projection~\eqref{eq:hardnet_aff} to meet the standard Lyapunov decrease condition for stability as follows:~\eqref{eq:hardnet_aff} is applied to a nominal dynamics model $\hat{f}$ to obtain $f^* \coloneqq \mathcal{P}(\hat{f})$ so that $\frac{dV}{dt} = \nabla V(x)^{\top} f^*(x) \leq -\alpha V(x)$ is guaranteed. Here, $\nabla V(x)^\top$ is interpreted as $A(x)$ in \eqref{eq:affine_constraints} and $-\alpha V(x)$ as $b_2(x)$. Then, supervised learning based on input--output data of $f$ is conducted on $f^*$. \cite{min2023} applies the same projection to the closed loop dynamics under controller setting. This facilitates the following statement: 
%
%
\begin{statement} \label{statement:01}
    If $\nabla V \neq 0$ except at $x = 0$, then the learned Lyapunov function guarantees stability for a surrogate system that is close to the true dynamics $f$ in the $\mathcal{L}_{\infty}$ sense, namely the learned dynamics $f^*$.
\end{statement}
This reduces the post-verification of certificate (i.e., satisfaction of the inequality constraint) to verifying $\nabla V \neq 0$. In Statement~\ref{statement:01}, \emph{closeness} is a physically interpretable quantity (same units as $f$), not a penalty weight--dependent value. By contrast, in soft--penalty training, violation magnitudes depend on arbitrary scalings of $V$ and on penalty weights, undermining physical interpretability as in Statement~\ref{statement:01}. The biggest benefit of applying HardNet is that the search space of the dynamics model is restricted to stable ones, without losing uniform approximator property.
%
%
\section{Problem Statement} \label{sec:prob_statement}
In Section~\ref{sec:preliminaries} we introduced $\delta$ISS notion on $\mathbb{R}^{n_x}\times\mathcal{U}$. In practice, learning and certification are carried out on a compact domain, so we fix a compact set $\mathcal{X}\subset\mathbb{R}^{n_x}$ and assume $\mathcal{U}$ compact. The dynamics $f:\mathcal X\times\mathcal U\mapsto\mathbb{R}^{n_x}$ is unknown, but we have data $D=\{(x_i,u_i,\dot x_i)\}_{i=1}^N$ of input-output measurements of $f$. The output map $h:\mathcal{X}\!\mapsto\!\mathbb{R}^{n_y}$ is known. Our goal is to jointly learn (i) a dynamics model $f^*$, (ii) an output injection $L^*$ so that the observer 
\begin{equation} \label{eq:observer_learned}
    \dot{z} = f^*(z, u) + L^*(z, u, y)
\end{equation}
%
%
%
is $\delta$ISS with respect to $(u,y)$. Because $h$ is known, we write the injection as $L^*(z,u,y)$ (rather than $L^*(z,u,y-h(z))$) and impose the consistency condition as $L^*(z, u, h(z)) \equiv 0$ for all $(z, u) \in \mathcal{X} \times \mathcal{U}$ so that $z(t)$ estimates the true state $x(t)$. To certify that \eqref{eq:observer_learned} is $\delta$ISS, we jointly learn a $\delta$ISS Lyapunov function $V:\mathcal X\times\mathcal X\to\mathbb{R}_{\ge 0}$ satisfying the following for some $\alpha,\varepsilon>0$ and $\rho_u,\rho_y\ge 0$.
\begin{equation} \label{eq:lyapunov_zero}
    V(z_1, z_2) = 0 \iff z_1 = z_2,
\end{equation}
\begin{equation} \label{eq:lyapunov_lb}
    V(z_1, z_2) \geq \varepsilon |z_1 - z_2|^2,
\end{equation}
\begin{align} \label{eq:lyapunov_decrease}
    \dot{V} &= \partial_{z_1}V^\top \left( f^*(z_1, u_1) + L^*(z_1, u_1, y_1) \right) \notag \\
    & \quad + \partial_{z_2}V^\top  \left( f^*(z_2, u_2) + L^*(z_2, u_2, y_2) \right) \\
    & \leq -\alpha V(z_1, z_2) + \rho_u|u_1 - u_2|^2 + \rho_y|y_1-y_2|^2. \notag
\end{align}
On a compact domain, these are sufficient conditions for the $\delta$ISS Lyapunov inequalities (cf. Definition~\ref{def:isslyapunov}) applied to the observer dynamics with inputs $(u,y)$. The next proposition yields $\mathcal L_\infty$ gains with respect to disturbances on $(u, y)$.
\begin{proposition} \label{prop:disturbance_gain}
    Under the decrease condition \eqref{eq:lyapunov_decrease}, $\forall t \ge 0$,
    \begin{align}
        \left|z_1(t) - z_2(t)\right| \leq \sqrt{V(0)/\varepsilon} \ & e^{-\frac{\alpha t}{2}} + \sqrt{\rho_u / \varepsilon \alpha} \ \|u_1 - u_2\|_{\infty} \notag \\
        & + \sqrt{\rho_y / \varepsilon \alpha} \ \|y_1 - y_2\|_{\infty}
    \end{align}
    where $V(0):=V\bigl(z_1(0),z_2(0)\bigr)$.
\end{proposition}
\begin{proof}
    See Appendix~\ref{appendix:a}.
\end{proof}
Note that the decrease condition in~\eqref{eq:lyapunov_decrease} is stated for the learned model $f^*$ rather than the (unknown) true dynamics $f$. Robustness to model error is analyzed later.

In summary, we jointly learn (i) the dynamics model $f^*$, (ii) an output injection $L^*$ to obtain a robust observer, and (iii) a $\delta$ISS Lyapunov function $V$ that certifies robustness in the $\delta$ISS sense.

For the two--subsystem interconnection in~\ref{eq:interconnected_system}, let us assume that datasets $\mathcal{D}_a = \{ (x_{a_i}, x_{b_i}, u_{a_i}, \dot{x}_{a_i}) \}_{i=1:N_a}$ and $\mathcal{D}_b = \{ (x_{a_i}, x_{b_i}, u_{b_i}, \dot{x}_{b_i}) \}_{i=1:N_b}$ are given for $f_a$ and $f_b$, respectively. Output maps $y_a = h_a(x_a)$ and $y_b = h_b(x_b)$ are known. We construct local observers (and $\delta$ISS Lyapunov functions $V_a,V_b$) analogously to the non-interconnected case: 
\begin{subequations} \label{eq:learned_distributed_observer}
\begin{align}
\dot{z}_a &= f_a^*(z_a, z_b, u_a) + L_a^*\bigl(z_a, z_b, u_a, y_a\bigr), \label{eq:learned_distributed_observer_a} \\
\dot{z}_b &= f_b^*(z_a, z_b, u_b) + L_b^*\bigl(z_a, z_b, u_b, y_b\bigr). \label{eq:learned_distributed_observer_b}
\end{align}
\end{subequations}
Here, $f_a^*$ and $f_b^*$ are learned dynamics models and $L_a^*$ and $L_b^*$ are the corresponding output injections that render each local observer systems $\delta$ISS. We impose the consistency conditions $L_a^*(z_a,z_b,u_a,h_a(z_a))\equiv 0$ and $L_b^*(z_a,z_b,u_b,h_b(z_b))\equiv 0$. For the learned local observer \eqref{eq:learned_distributed_observer_a}, suppose $V_a$ is trained to satisfy 
\begin{equation}
    V_a(z_a^1,z_a^2)\;\ge\;\varepsilon_a\left|z_a^1-z_a^2\right|^2,
\end{equation}
\begin{align}
    \dot V_a \!\leq\! -\alpha_a\! V_a \!+\! \rho_{u_a}\!\left|u_a^1 \!-\! u_a^2\right|^2 \!+\! \rho_{y_a}\!\left|y_a^1 \!-\! y_a^2\right|^2 \!+\! \rho_{z_b}\!\left|z_b^1 \!-\! z_b^2\right|^2,
\end{align}
for $\alpha_a,\varepsilon_a>0$, $\rho_{u_a},\rho_{y_a},\rho_{z_b}\ge 0$. Then, by Proposition~\ref{prop:disturbance_gain}, the $\mathcal L_\infty$ gain from $z_b^1-z_b^2$ to $z_a^1-z_a^2$ is $\sqrt{\rho_{z_b}/(\varepsilon_a\alpha_a)}$. Analogously, for \eqref{eq:learned_distributed_observer_b} we obtain the gain $\sqrt{\rho_{z_a}/(\varepsilon_b\alpha_b)}$. If 
\begin{equation} \label{eq:small_gain}
    \sqrt{\frac{\rho_{z_b}}{\varepsilon_a\alpha_a}\cdot\frac{\rho_{z_a}}{\varepsilon_b\alpha_b}} \;<\; 1,
\end{equation}
then, by Theorem~\ref{thm:interconnected}, the interconnected observer \eqref{eq:learned_distributed_observer} is $\delta$ISS with respect to the inputs $(u_a, u_b, y_a, y_b)$. In summary, each local observer is learned as in the single--system case, without any knowledge on the other subsystem. The only information shared during construction of local observers is the small--gain condition \eqref{eq:small_gain} for parameter setting.
%
\section{Joint Learning of Dynamics, Observer, and $\delta$ISS Lyapunov Function} \label{sec:architecture}
In this section, we present and analyze the architecture for constructing a robust full--state observer. The same approach applies to the construction of local observers for interconnected systems.
\subsection{Dynamics, output injection, and $\delta$ISS Lyapunov function}
Prior work~\cite{kolter2019,min2023} on stability enforces the Lyapunov decrease condition by projecting a nominal dynamics model through~\eqref{eq:hardnet_aff}. Here, we present how projection~\eqref{eq:hardnet_aff} can be applied to~\eqref{eq:lyapunov_decrease} and lead to arguments such as Statement~\ref{statement:01}.

In inequality~\eqref{eq:lyapunov_decrease}, the left--hand side is affine in both $L^*$ and $f^*$. In principle one could enforce the constraint by correcting either component. In our approach we adjust $L^*$ rather than $f^*$. The rationale is straightforward: since $f^*$ is directly trained on $\mathcal D$, we consider $L^*$ as the adjustable component for enforcing the inequality, not $f^*$. Accordingly, we model the dynamics with an unconstrained neural network $f^*:\mathbb{R}^{n_x}\times\mathbb{R}^{n_u}\to\mathbb{R}^{n_x}$ trained on $\mathcal D$. We define the output injection as 
\begin{equation} \label{eq:lstar}
    L^*(x, u, y) = g_{L}(x, u, y) - g_{L}(x, u, h(x))
\end{equation}
which enforces the consistency condition $L^*(z,u,h(z))\equiv 0$, where $g_L:\mathbb{R}^{n_x+n_u+n_y}\to\mathbb{R}^{n_x}$ is an unconstrained network. For the $\delta$ISS Lyapunov function, we set as follows to meet~\eqref{eq:lyapunov_zero} and inequality~\eqref{eq:lyapunov_lb}.
%
\begin{equation}
    V(x_1, x_2) = \sigma\!\left( g_V(x_1, x_2) - g_V(x_1, x_1) \right) + \varepsilon_V \! \left| x_1 - x_2 \right|^2,
\end{equation}
for $\varepsilon_V > 0$, $g_V: \mathbb{R}^{2n_x} \mapsto \mathbb{R}$ an unconstrained network, and a nonnegative map $\sigma: \mathbb{R} \mapsto \mathbb{R}_{\geq 0}$ such that $\sigma(s) = 0$ if $s \leq 0$. It is noteworthy that the parametrization must be capable of representing the \emph{gradient} of $V$ accurately, thus requiring $\sigma$ and the activation functions for $g_V$ to be at least $C^1$. An example of such $\sigma$ is the smoothed ReLU function in~\cite{kolter2019}.
\subsection{Decreasing condition of $\delta$ISS Lyapunov function}
We present a HardNet variant that operates on the \emph{pairwise} $\delta$ISS Lyapunov decrease constraint rather than on the dynamics itself. A schematic illustration of the two steps (projection to enforce $\delta$ISS Lyapunov function to decrease, and a subsequent step for consistency condition) is shown in Figure~\ref{fig:diagram}. Suppose there exists an ideal injection $L$ that renders $\dot z=f^*(z,u)+L(z,u,y)$ $\delta$ISS. If one approximates $L$ by an unconstrained network $\hat L$, it is natural to restrict $\hat{L}$ to maintain the Lyapunov decrease condition~\eqref{eq:lyapunov_decrease}. The difficulty is that~\eqref{eq:lyapunov_decrease} couples \emph{pairs} of points $(z_1,u_1,y_1)$ and $(z_2,u_2,y_2)$, unlike the pointwise affine constraint in~\eqref{eq:affine_constraints}. To apply projection under the pairwise structure of \eqref{eq:lyapunov_decrease}, we embed the constraint by defining
\begin{equation}
    \tilde{L}(x_1, u_1, y_1, x_2, u_2, y_2) \! \coloneqq \!
    \begin{bmatrix}
        \hat{L}_1(x_1, u_1, y_1, x_2, u_2, y_2) \\
        \hat{L}_2(x_1, u_1, y_1, x_2, u_2, y_2)
    \end{bmatrix}, 
\end{equation}
where $\hat L_1,\hat L_2:\mathbb{R}^{2(n_x+n_u+n_y)} \!\!\to\!\!\mathbb{R}^{n_x}$ are unconstrained networks. Let $\nabla V(z_1,z_2)=\left[\partial_{z_1}V(z_1,z_2)^\top, \ \partial_{z_2}V(z_1,z_2)^\top\right]\in\mathbb{R}^{1\times 2n_x}$. Consider the following pointwise inequality in embedded space,
\begin{align} \label{eq:embedded_ineq}
    & \nabla V(z_1, z_2) \Big( \begin{bmatrix}
        f^*(x_1, u_1) \\
        f^*(x_2, u_2)
    \end{bmatrix} + \tilde{L}(z_1, u_1, y_1, z_2, u_2, y_2) \Big) \notag \\
    & \leq -\alpha V(z_1, z_2) + \rho_u|u_1 - u_2|^2 + \rho_y|y_1-y_2|^2.
\end{align}
If $\hat L_1(\cdot)\equiv L(z_1,u_1,y_1)$ and $\hat L_2(\cdot)\equiv L(z_2,u_2,y_2)$, then~\eqref{eq:embedded_ineq} reduces to~\eqref{eq:lyapunov_decrease}. Writing $\chi_i:=(z_i,u_i,y_i)$, $i \in \{1, 2\}$, inequality~\eqref{eq:embedded_ineq} is equivalent to
\begin{align} \label{eq:lyapunov_decrease_short}
    \nabla V(z_1, z_2) \tilde{L}(\chi_1, \chi_2) \leq \Lambda(\chi_1, \chi_2)
\end{align}
where $\Lambda(\chi_1, \chi_2) \coloneqq -\nabla V(z_1, z_2) \begin{bmatrix}
    f^*(x_1, u_1) \\
    f^*(x_2, u_2)
\end{bmatrix} + \rho_u|u_1 - u_2|^2 + \rho_y|y_1-y_2|^2$. Applying projection~\eqref{eq:hardnet_aff} to $\tilde L$ yields
\begin{align} \label{eq:L_projection}
    &\tilde{L}^*(\chi_1, \chi_2) \coloneqq \tilde{L}(\chi_1, \chi_2) \\
    &- \nabla V(z_1, z_2)^\top \frac{\text{ReLU}(\nabla V(z_1, z_2) \tilde{L}(\chi_1, \chi_2) - \Lambda(\chi_1, \chi_2))}{\|\nabla V(z_1, z_2)\|^2}, \notag
\end{align}
which guarantees~\eqref{eq:lyapunov_decrease_short} provided $\nabla V(z_1,z_2)\neq 0$ for $z_1 \neq z_2$. Denoting the components by $\tilde{L}^*(\chi_1, \chi_2) = \begin{bmatrix}
    L_1^*(\chi_1, \chi_2)^\top ,\; L_2^*(\chi_1, \chi_2)^\top
\end{bmatrix}^\top$, the interpretation of applying projection~\eqref{eq:hardnet_aff} in embedded space is as follows. First, the resulting coupled system
\begin{align} \label{eq:coupled_observer}
    \dot{z}_1 &= f^*(z_1, u_1) + L_1^*(\chi_1, \chi_2), \notag \\
    \dot{z}_2 &= f^*(z_2, u_2) + L_2^*(\chi_1, \chi_2)
\end{align}
satisfies 
\begin{align} \label{eq:lyapunov_embedded}
    \dot{V}(z_1, z_2) &= \partial_{z_1}V^\top \left( f^*(z_1, u_1) + L_1^*(\chi_1, \chi_2) \right) \notag \\
    & \quad + \partial_{z_2}V^\top \left( f^*(z_2, u_2) + L_2^*(\chi_1, \chi_2) \right) \\
    & \leq -\alpha V(z_1, z_2) + \rho_u|u_1 - u_2|^2 + \rho_y|y_1-y_2|^2, \notag
\end{align}
implying that $z_1$ and $z_2$ converge to each other as in Proposition~\ref{prop:disturbance_gain}. Moreover, any continuous pair $(L_1^*,L_2^*)$ satisfying~\eqref{eq:lyapunov_embedded} can be approximated arbitrarily well as $\hat L_1,\hat L_2$ are universal approximators. Thus, for an ideal injection $L$, the special case $L_1^*(\chi_1,\chi_2)\equiv L(\chi_1)$ and $L_2^*(\chi_1,\chi_2)\equiv L(\chi_2)$ lies in the search space. Conversely, if $L^*$ in~\eqref{eq:lstar} satisfies $L_1^*(\chi_1,\chi_2)\equiv L^*(\chi_1)$ and $L_2^*(\chi_1,\chi_2)\equiv L^*(\chi_2)$, then~\eqref{eq:lyapunov_embedded} implies that $L^*$ meets the inequality~\eqref{eq:lyapunov_decrease}, becoming the desired output injection.
\begin{figure}
    \centering
    \includegraphics[width=0.55\linewidth]{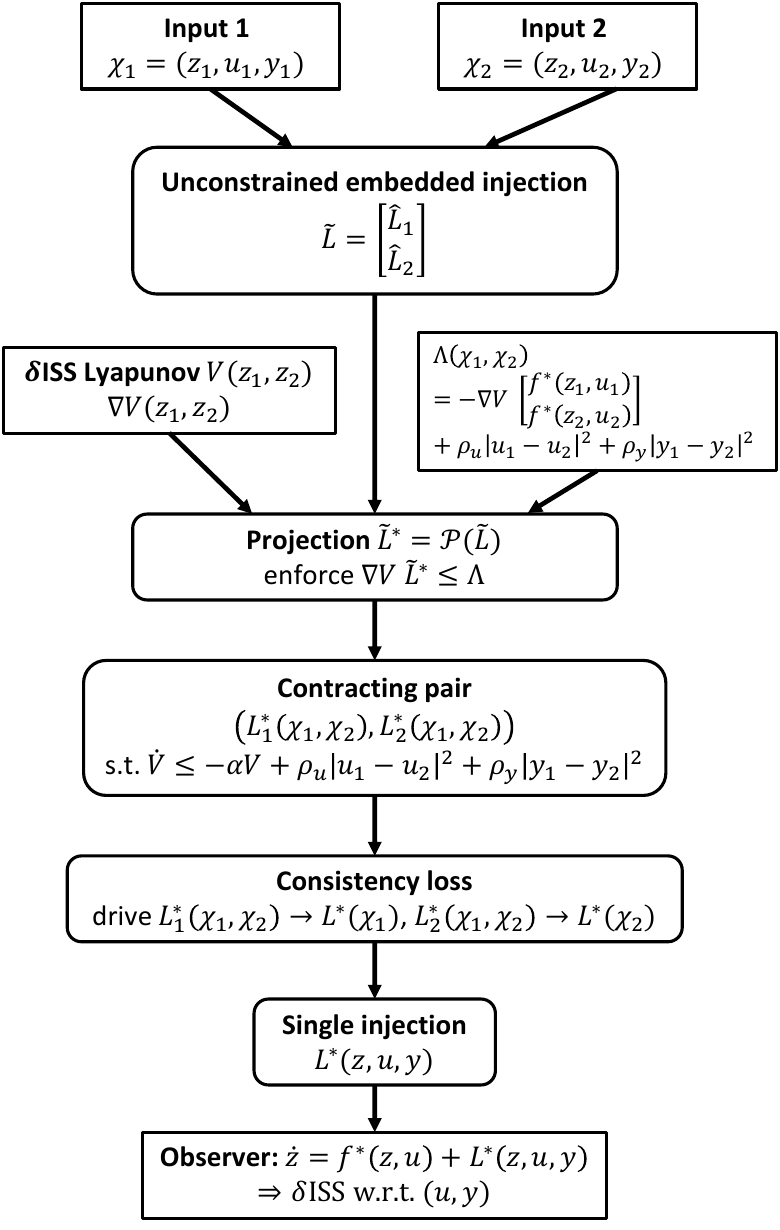}
    \caption{Two-step procedure. An unconstrained embedded injection $\tilde{L}$ is projected to satisfy $\nabla V \tilde{L} \leq \Lambda$, yielding a contracting pair $(L_1^*, L_2^*)$. A consistency loss then drives $L_1^*, L_2^* \rightarrow L^*$.}
    \label{fig:diagram}
\end{figure}
\subsection{Loss function}
When jointly learning $f^*$, $V$, $L^*$, $\hat{L}_1$, and $\hat{L}_2$, we train $f^*$ on $\mathcal D$, while $(L^*, \hat{L}_1, \hat{L}_2)$ are trained so that the projected components $L_1^*,L_2^*$ obtained from~\eqref{eq:L_projection} satisfy $L_1^*(\chi_1,\chi_2)\equiv L^*(\chi_1)$ and $L_2^*(\chi_1,\chi_2)\equiv L^*(\chi_2)$. To avoid numerical instability when $\|\nabla V\|$ is small in~\eqref{eq:L_projection}, we clamp the denominator as $\max\!\big(\varepsilon_{\mathrm{proj}}, \|\nabla V\|^2\big)$ with a small constant $\varepsilon_{\mathrm{proj}} > 0$ as in~\cite{min2023}. We build an unlabeled nominal dataset $\mathcal{D}_{\chi} = \{ \chi_i = (x_i, u_i, y_i) \}_{i=1:N_{\chi}}$ for $(L^*, \hat{L}_1, \hat{L}_2)$ sampled from $\mathcal{X} \times \mathcal{U} \times \mathcal{Y}$. The loss components are: 
\begin{align}
    \mathcal{L}_{f^*} &= \frac{1}{N} \sum\limits_{(x, u, \dot{x}) \in \mathcal{D}} | f^*(x, u) \!-\! \dot{x} |^2, \\
    \mathcal{L}_{L^*} &= \frac{1}{N_{\chi}^2} \sum\limits_{\chi_1, \chi_2 \in \mathcal{D}_{\chi}} \left| \tilde{L}^*(\chi_1,\chi_2) \!-\! \begin{bmatrix} L^*(\chi_1) \\ L^*(\chi_2) \end{bmatrix} \right|^2. \notag \\
\end{align}
We then minimize the total loss
\begin{equation} \label{eq:total_loss}
    \mathcal{L} = \lambda_f \mathcal{L}_{f^*} + \lambda_{L^*} \mathcal{L}_{L^*}, 
\end{equation}
with nonnegative weights $\lambda_f$ and $\lambda_{L^*}$.


Minimizing~\eqref{eq:total_loss} drives the embedded pair $(L_1^*,L_2^*)$ that is guaranteed to meet~\eqref{eq:lyapunov_decrease_short} by construction toward a single injection $L^*$ with the consistency condition. In this procedure, $L^*$ and $(L_1^*,L_2^*)$ maintain the uniform approximator property. Conceptually, prior work on stability~\cite{kolter2019,min2023} restricts the dynamics class to models that satisfy the Lyapunov decrease condition and fits them to meet $f$. Similarly, we restrict the embedded injections $(L_1^*,L_2^*)$ to contracting ones (satisfying~\eqref{eq:lyapunov_decrease_short}) and fit them to meet the consistency condition. To elaborate, let us write $\Updelta_{L_1}(\chi_1, \chi_2) \coloneqq L^*(\chi_1) - L_1^*(\chi_1, \chi_2)$ and $\Updelta_{L_2}(\chi_1, \chi_2) \coloneqq L^*(\chi_2) - L_2^*(\chi_1, \chi_2)$.
%
Note that the inequality~\eqref{eq:lyapunov_embedded} can be rewritten as 
\begin{align} \label{eq:lyapunov_decrease_delta}
    \dot{V} &= \partial_{z_1}V^\top \! \left( f^*(z_1, u_1) + L^*(\chi_1) - \Updelta_{L_1}(\chi_1, \chi_2) \right) \notag \\
    & \quad + \partial_{z_2}V^\top \! \left( f^*(z_2, u_2) + L^*(\chi_2) - \Updelta_{L_2}(\chi_1, \chi_2) \right) \\
    & \leq -\alpha V(z_1, z_2) + \rho_u|u_1 - u_2|^2 + \rho_y|y_1-y_2|^2. \notag
\end{align}
Here, we view $\Updelta_{L_1}$ and $\Updelta_{L_2}$ as model errors that are minimized during the training process, and let $\Updelta_{L} \;\coloneqq\; \sup\limits_{\chi_1,\chi_2 \in \mathcal{X}\times\mathcal{U}\times\mathcal{Y}} \max\!\big\{|\Updelta_{L_1}(\chi_1,\chi_2)|,\ |\Updelta_{L_2}(\chi_1,\chi_2)|\big\}$. Now, recall that the $\delta$ISS property of system~\eqref{eq:observer_learned} means \emph{pairwise convergence}: trajectories $z_1$ and $z_2$ generated from inputs $(u_1,y_1)$ and $(u_2,y_2)$ that are close (in $\|\cdot\|_{\infty}$ sense) converge toward each other; the same interpretation applies to $z_1$ and $z_2$ in coupled system~\eqref{eq:coupled_observer}. Thus,~\eqref{eq:lyapunov_decrease_delta} implies that the pairwise convergence condition for the actual observer system~\eqref{eq:observer_learned} holds for a surrogate that is close to it (within the model--error bound $\Updelta_L$). Then the $\delta$ISS version of Statement~\ref{statement:01} is: 
\begin{statement} \label{statement:02}
If $\nabla V \neq 0$ except at $z_1 = z_2$, then the learned $\delta$ISS Lyapunov function guarantees pairwise convergence for a surrogate system that is close (in $\mathcal{L}_{\infty}$ sense, up to $\Updelta_L$) to the true observer system~\eqref{eq:observer_learned}, namely the coupled system~\eqref{eq:coupled_observer}.
\end{statement}
%
%
\subsection{Performance guarantees}
We analyze the performance of the learned observer~\eqref{eq:observer_learned}.
\begin{theorem} \label{thm:01}
    Let $M_i = \max\limits_{\mathcal{X} \times \mathcal{X}} \|\partial_{z_i} V(z_1, z_2)\|$, $i \in \{ 1, 2 \}$, $\Updelta_{f} \coloneqq \sup\limits_{\mathcal{X} \times \mathcal{U}} \| f(z, u) - f^*(z, u) \|$, and $\Updelta \coloneqq M_1\left(\Updelta_f + \Delta_{L_1}\right) + M_2\Delta_{L_2}$. Assume $x(t), z(t) \in \mathcal{X}$ and $\nabla V \neq 0$ if $z_1 \neq z_2$. 
    If an input $u$ generates true state $x(t)$ and output $y = h(x)$ from~\eqref{eq:system}, and the learned observer~\eqref{eq:observer_learned} generates estimate $z(t)$ from input $u'$ and output $y'$, then the following estimate holds for $V(t) = V(x(t), z(t))$ and all $t \geq 0$.
    \begin{align} \label{eq:performance_bound}
        |x(t) - z(t)| &\leq \sqrt{\frac{V(0)}{\varepsilon}} e^{-\frac{\alpha t}{2}} \!+\! \sqrt{\frac{\rho_u}{\alpha \varepsilon}} \left\|u \!-\! u'\right\|_{\infty} \notag \\
        & \quad +\! \sqrt{\frac{\rho_y}{\alpha \varepsilon}} \left\|y \!-\! y'\right\|_{\infty} + \sqrt{\frac{\Updelta}{\alpha \varepsilon}}. 
    \end{align}
    \hfill $\square$
\end{theorem}
\begin{proof}
    Let $\chi' = (z, u', y')$. From $L^*(x,u,h(x))=0$, 
    \begin{align}
        &\dot V = \partial_{z_1}V^\top\!\left[f(x,u)+L^*(x,u,h(x))\right] \notag \\
        &\qquad \qquad + \partial_{z_2}V^\top\!\left[f^*(z,u')+L^*(z,u',y')\right] \notag \\
        &= \partial_{z_1}V^\top\!\left[f^*(x,u)+L_1^*(\chi,\chi')+\Delta_f(x,u)+\Delta_{L_1}(\chi,\chi')\right] \notag \\
        &\quad + \partial_{z_2}V^\top\!\left[f^*(z,u')+L_2^*(\chi,\chi')+\Delta_{L_2}(\chi,\chi')\right] \notag \\
        &\le \partial_{z_1}V^\top\!\!\left[f^*(x,u)\!+\!L_1^*(\chi,\chi')\right] \!+\! \partial_{z_2}V^\top\!\!\left[f^*(z,u') \! + \! L_2^*(\chi,\chi')\right] \notag \\
        & \quad + M_1\!\left(|\Delta_f| \! + \! |\Delta_{L_1}|\right) \!+\!M_2|\Delta_{L_2}| \notag \\
        &\le -\alpha V(x,z) + \rho_u\lvert u-u'\rvert^2 + \rho_y\lvert y-y'\rvert^2 \notag \\
        &\quad + M_1\!\left(|\Delta_f| \! + \! |\Delta_{L_1}|\right)+M_2|\Delta_{L_2}|.
    \end{align}
    Then~\eqref{eq:performance_bound} follows along the same lines of the proof of Proposition~\ref{prop:disturbance_gain}.
\end{proof}
\begin{remark}
    In practice, we add an auxiliary penalty $\mathcal{L}_{\nabla V} = \frac{1}{N_{\chi}^2} \sum\limits_{x_1, x_2} \text{ReLU}\left(\varepsilon_{proj} - \|\nabla V(x_1, x_2)\|^2\right)$ to discourage small gradient norms away from the diagonal $S = \{(z_1,z_2): z_1=z_2\}$. This encourages $\|\nabla V\|^2 \ge \varepsilon_{\mathrm{proj}}$ outside a small neighborhood of $S$. On that region, the clamped projection in \eqref{eq:L_projection} guarantees the decrease condition \eqref{eq:lyapunov_decrease_short}. Thus, certificate verification~\cite{dawson2023} reduces to finding $\Updelta$ and $r>0$ such that $|z_1-z_2|>r \Rightarrow \|\nabla V(z_1,z_2)\|^2 \ge \varepsilon_{\mathrm{proj}}$, similar to~\cite{min2023}. Then the performance bound \eqref{eq:performance_bound} can be written as $\lvert x(t)-z(t)\rvert \;\le\;\max\!\left\{\, r,\;\text{(RHS of~\eqref{eq:performance_bound})}\right\}$.
\end{remark}
\begin{remark}

    Under the assumptions of Theorem~\ref{thm:01}, for any observer trajectories driven by $(u_1,y_1)$ and $(u_2,y_2)$, 
    \[
    \begin{aligned}
    |z_1(t)-z_2(t)|
    &\le \max\Bigl\{\, r,\;
    \sqrt{\tfrac{\rho_u}{\alpha \varepsilon}}\;\|u_1-u_2\|_{\infty} \\
    &+ \sqrt{\tfrac{\rho_y}{\alpha \varepsilon}}\;\|y_1-y_2\|_{\infty} + \sqrt{\tfrac{M_1\Updelta_{L_1}+M_2\Updelta_{L_2}}{\alpha \varepsilon}}
    \Bigr\},
    \end{aligned}
    \]
    which implies practical $\delta$ISS~\cite{jiang1994}. Consequently, if local observers are designed for the interconnection in Section~\ref{sec:prob_statement}, then each local observer is practically $\delta$ISS. This implies that the interconnected observer~\eqref{eq:distributed_observer} is practically $\delta$ISS as well. Moreover, since the practical $\delta$ISS residuals (ultimate bounds) of the local observers have already been derived above, the residual for the interconnected observer \eqref{eq:distributed_observer} follows directly from the practical small-gain composition in~\cite{jiang1994}. As all ingredients are in place, we omit the explicit construction for brevity. \hfill $\square$
\end{remark}
%
\section{Numerical Demonstration} \label{sec:demo}
We validate the proposed framework on three scenarios of increasing complexity: (i) the uncontrolled Lorenz attractor, (ii) bicycle path--following under open--loop input, and (iii) an interconnected FitzHugh--Nagumo system~\cite{rankovic2011}. All models ($f^*$, $g_L$, $g_V$, $\hat L_1$, $\hat L_2$) are 3-layer MLPs (128–256–128, $\tanh$). Loss weights are $\lambda_f{=}1$, $\lambda_{L^*}{=}10$, $\lambda_{\nabla V}{=}100$. For $g_V$, we add $\tanh$ activation function multiplied by $100$ at the output layer to scale the value of the $\delta$ISS Lyapunov function. We use the smoothed ReLU function with threshold $d{=}10^{-3}$~\cite{kolter2019} for $\sigma$. Training uses Adam (lr $10^{-3}$, batch size $1000$) for $1000$ epochs with $\varepsilon_V{=}0.2$, $\alpha{=}1.0$, and $\varepsilon_{\mathrm{proj}}{=}10^{-4}$. Integrations use RK4 with step size $0.005$. 
%
%
\subsection{Uncontrolled Lorenz attractor}
We start with the uncontrolled Lorenz attractor: 
\begin{align}
    \dot{x}_1 &= \sigma (x_2 - x_1), \notag \\
    \dot{x}_2 &= \rho x_1 - x_2 - x_1 x_3, \\
    \dot{x}_3 &= -\beta x_3 + x_1 x_2, \notag
\end{align}
with $\sigma = 10$, $\rho = 28$, and $\beta = 8/3$. We measure $y = x_1$. The training set $\mathcal{D}$ is uniformly sampled from $\mathcal{X} = [-25, 25] \times [-25, 25] \times [0, 50]$ with $N = 10^5$, and $\mathcal{D}_{\chi}$ is sampled from $\mathcal{X} \times \mathcal{Y}$ with $\mathcal{Y} = [-25, 25]$ and $N_{\chi} = 10^6$. We use $\rho_y = 2.0$ for the disturbance gain. We compare our learned observer against the Deep--KKL observer of~\cite{peralez2024deep}, with forgetting factor $a = 0.99$.
\begin{figure}
    \centering
    \includegraphics[width=1\linewidth]{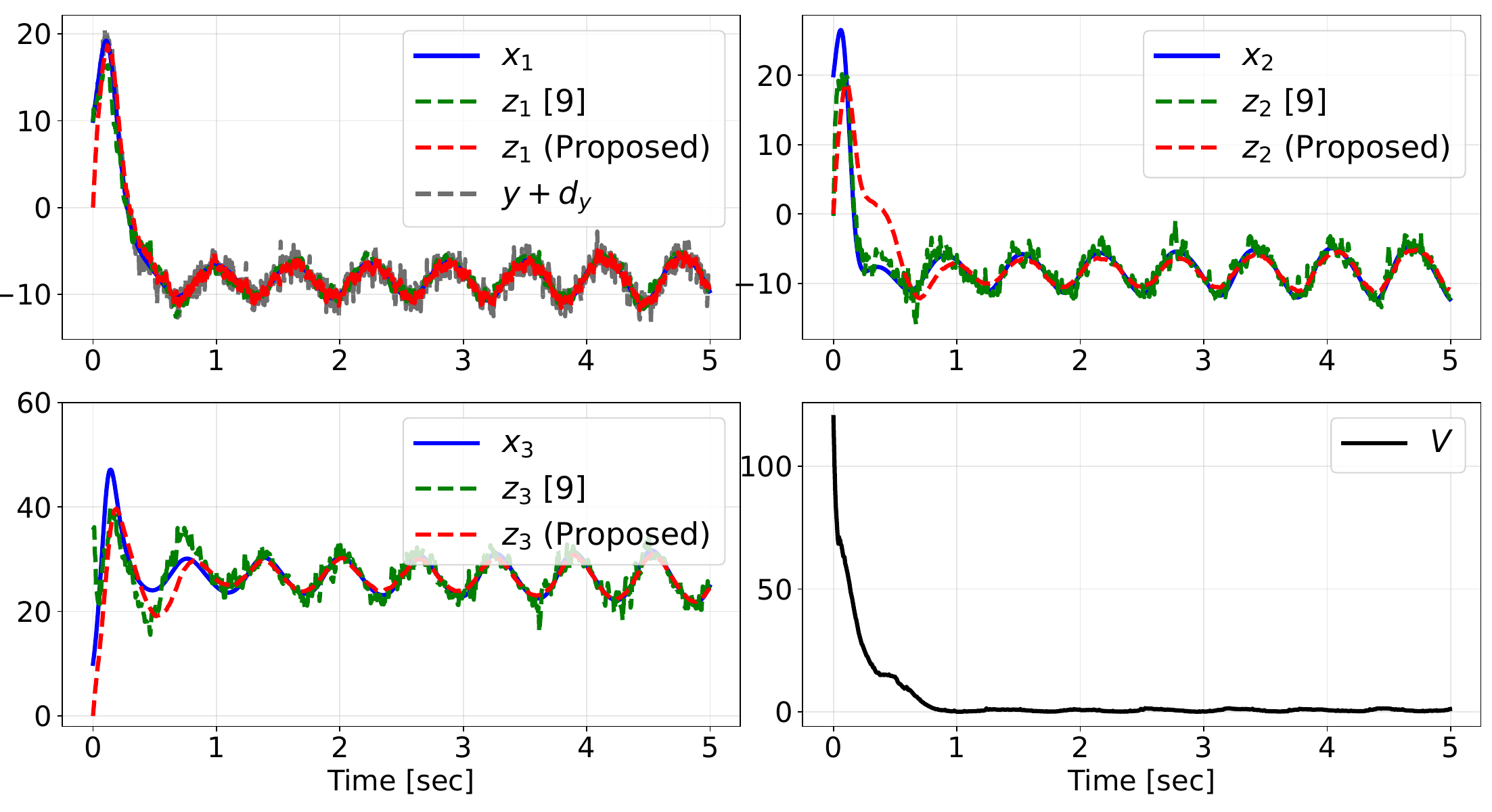}
    \caption{Observer results for Lorenz oscillator with noisy measurement $y + d_y$ with $d_y \sim \mathcal{N}(0,1)$.}
    \label{fig:lorenz_noise}
\end{figure}
Figure~\ref{fig:lorenz_noise} illustrates the performance under additive measurement noise $d_y \sim \mathcal{N}(0,1)$. The true initial state is $x(0) \!=\! [10,20,10]$, while we initialize our observer as $z(0) \!=\! [0,0,0]$. Note that~\cite{peralez2024deep} includes an initialization network $\phi$ that estimates the latent initial state from the first measurement.

Over $t \in [0,5]$ under noise--free setting, the root--mean--square error (RMSE) of our approach is $3.398$ versus $3.142$ for Deep--KKL, which is slightly higher due to the different initial value setting. Restricting to $[1,5]$ (after transients decay), our approach achieves an RMSE of $0.9841$ compared with $1.647$ for Deep--KKL, indicating better steady--state accuracy. To assess robustness, we inject zero--mean Gaussian noise into the measurement, i.e., both observers receive $y+d_y$ with $d_y\sim\mathcal{N}(0,1)$. Over $[0,5]$, our method attains an RMSE of $4.007$ while Deep--KKL records $5.284$. Over $[1,5]$, the respective RMSEs are $2.475$ and $2.587$, showing improved noise robustness for the proposed observer.
\subsection{Bicycle path--following with an open--loop input}
The second example is a nonautonomous system with non--affine dependency on the input. Prior results~\cite{miao2023learning,peralez2024deep} on KKL observers either focus on autonomous systems or require model knowledge for input handling; a model--free treatment of non--affine inputs was not shown. A bicycle of constant speed $v$ follows a path of constant curvature $\kappa$. Let $d_e$ be the lateral deviation and $\theta_e$ the heading error: 
\begin{align}
    \dot{d}_e &= v \sin(\theta_e), \notag \\
    \dot{\theta}_e &= \frac{v}{L} \tan(u) - \frac{v \kappa \cos(\theta_e)}{1 - \kappa d_e},
\end{align}
with $\kappa = 1$, $v = 6$, $L = 1$. We measure $y = d_e$. As a nominal open--loop input, we use $u(t) = \arctan\!\big(\bar{q} + \varepsilon \sin(\omega t)\big)$ with $\bar{q}=1$, $\varepsilon = 0.2$, and $\omega = 0.5$. The training set $\mathcal{D}$ is uniformly sampled from $\mathcal{X}\times\mathcal{U}$ with $\mathcal{X}=[-0.8,0.8]\times[-0.8,0.8]$, $\mathcal{U}=[-0.4\pi,\,0.4\pi]$, and $N=2\times 10^4$. $\mathcal{D}_\chi$ is sampled from $\mathcal{X} \times \mathcal{U} \times \mathcal{Y}$ with $\mathcal{Y} = [-0.8,0.8]$ and $N_\chi=10^5$. We use $\rho_u = \rho_y = 5$ for the disturbance gains. The true initial state is randomly sampled from $\mathcal{X}$, and we initialize the observer at $z(0)=[0.8,\,0.8]$ (a corner of the training domain).
\begin{figure}
    \centering
    \includegraphics[width=1\linewidth]{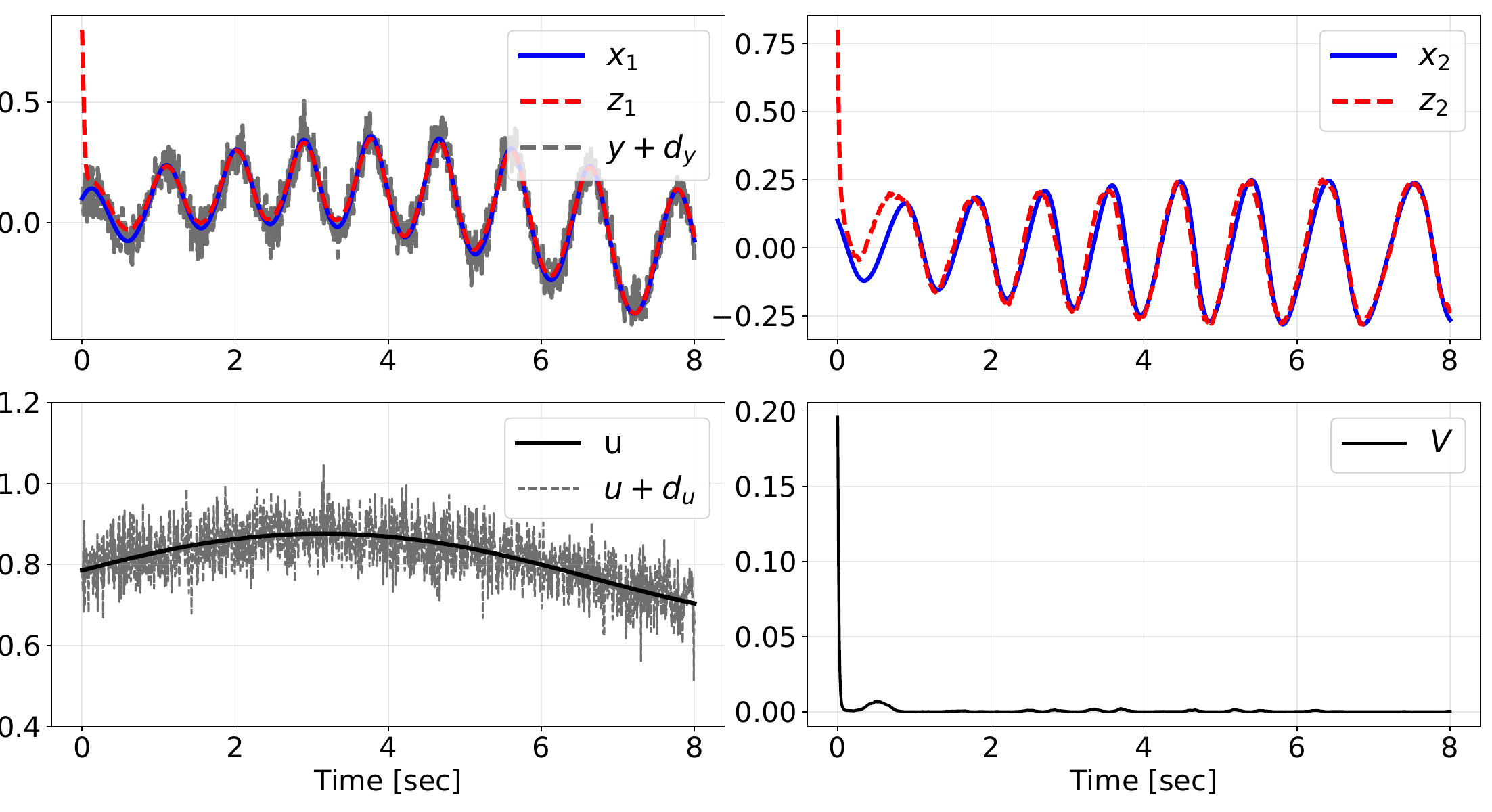}
    \caption{Observer results for bicycle path--following dynamics with noises $d_y \sim \mathcal{N}(0, 0.05)$ and $d_u \sim \mathcal{N}(0, 0.05)$ entering $y$ and $u$, respectively.}
    \label{fig:bicycle_noise}
\end{figure}

Under noise--free setting, our learned observer attains RMSE of $0.0592$ over $t\in[0,8]$, and $0.0137$ over stabilized window $[1,8]$. To assess robustness, we provide $y+d_y$ and $u+d_u$ to the observer, where $d_y\sim\mathcal{N}(0,0.05)$ and $d_u\sim\mathcal{N}(0,0.05)$. Figure~\ref{fig:bicycle_noise} illustrates the performance under measurement and input noises. On $t\in[0,8]$, the RMSE is $0.0808$, and $0.0587$ on $[1,8]$.
\subsection{Interconnected FitzHugh-Nagumo system}
Finally, we validate the distributed observer on an interconnected FitzHugh--Nagumo (FHN) pair~\cite{rankovic2011}. Two coupled neuronal subsystems receive inputs $u_a$ and $u_b$: 
\begin{align}
    \dot{x}_{a,1} &= x_{a,1} \!-\! \frac{1}{3}x_{a,1}^3 \!-\! x_{a,2} \!+\! \kappa_a u_a \!+\! R_a I_a \!+\! c \mathcal{F}(x_{a,1}, x_{b,1}), \notag \\
    \dot{x}_{a,2} &= \frac{1}{\tau}(x_{a,1} + a - b x_{a,2}), \\
    \dot{x}_{b,1} &= x_{b,1} \!-\! \frac{1}{3}x_{b,1}^3 \!-\! x_{b,2} \!+\! \kappa_b u_b \!+\! R_b I_b \!+\! c \mathcal{F}(x_{b,1}, x_{a,1}), \notag \\
    \dot{x}_{b,2} &= \frac{1}{\tau}(x_{b,1} + a - b x_{b,2}), \notag
\end{align}
where $\mathcal{F}(x_i, x_j) = -\frac{x_i - V_s}{1 + \exp(-k(x_j - \theta_s))} - \frac{V_s}{1 + \exp(k \theta_s)}$~\cite{rankovic2011} with parameters $a \!=\! 0.7$, $b \!=\! 0.8$, $\tau \!=\! 12.5$, $R_a \!=\! R_b \!=\! 1$, $I_a \!=\! 0.8$, $I_b \!=\! 1.6$, $\kappa_a \!=\! 0.1$, $\kappa_b \!=\! 0.2$, $c \!=\! 1$, $V_s \!=\! 1.5$, $k \!=\! 6$, and $\theta_s \!=\! -0.25$. Each subsystem outputs its membrane potential $y_a \!=\! x_{a,1}$ and $y_b \!=\! x_{b,1}$, where a common broadcast input $u(t) \!=\! u_a(t) \!=\! u_b(t) \!=\! \cos t$ is considered. We construct distributed observers via the two--step procedure in Section~\ref{subsec:robust_observers}. Local observers are trained independently from $\mathcal{D}_a = \{(x_{a}, x_{b,1}, u_a, \dot{x}_a)\}$ and $\mathcal{D}_b = \{(x_{a,1}, x_b, u_b, \dot{x}_b)\}$ drawn from $\mathcal{X} \times \mathcal{U} = [-2.5, 2.5]^4$ with $N_a \!=\! N_b \!=\! 2\times10^4$. $\mathcal{D}_{\chi_a}$ and $\mathcal{D}_{\chi_b}$ are drawn from $\mathcal{X} \times \mathcal{U} \times \mathcal{Y}$ with $\mathcal{Y} \!=\! [-2.5,2.5]$ and $N_{\chi_a} \!=\! N_{\chi_b}=10^5$. For disturbance gains, we set $\rho_{y_a} \!=\! \rho_{y_b} \!=\! \rho_{u_a} \!=\! \rho_{u_b} \!=\! 2$ and $\rho_{z_a} \!=\! \rho_{z_b} \!=\! 0.15$ so that the small--gain condition~\eqref{eq:small_gain} holds. The true initial state is randomly sampled from $\mathcal{X}$ for the both subsystems, and we initialize the observer at $z_a(0) \!=\! z_b(0) \!=\! [0,0]$.
\begin{figure}
    \centering
    \includegraphics[width=1\linewidth]{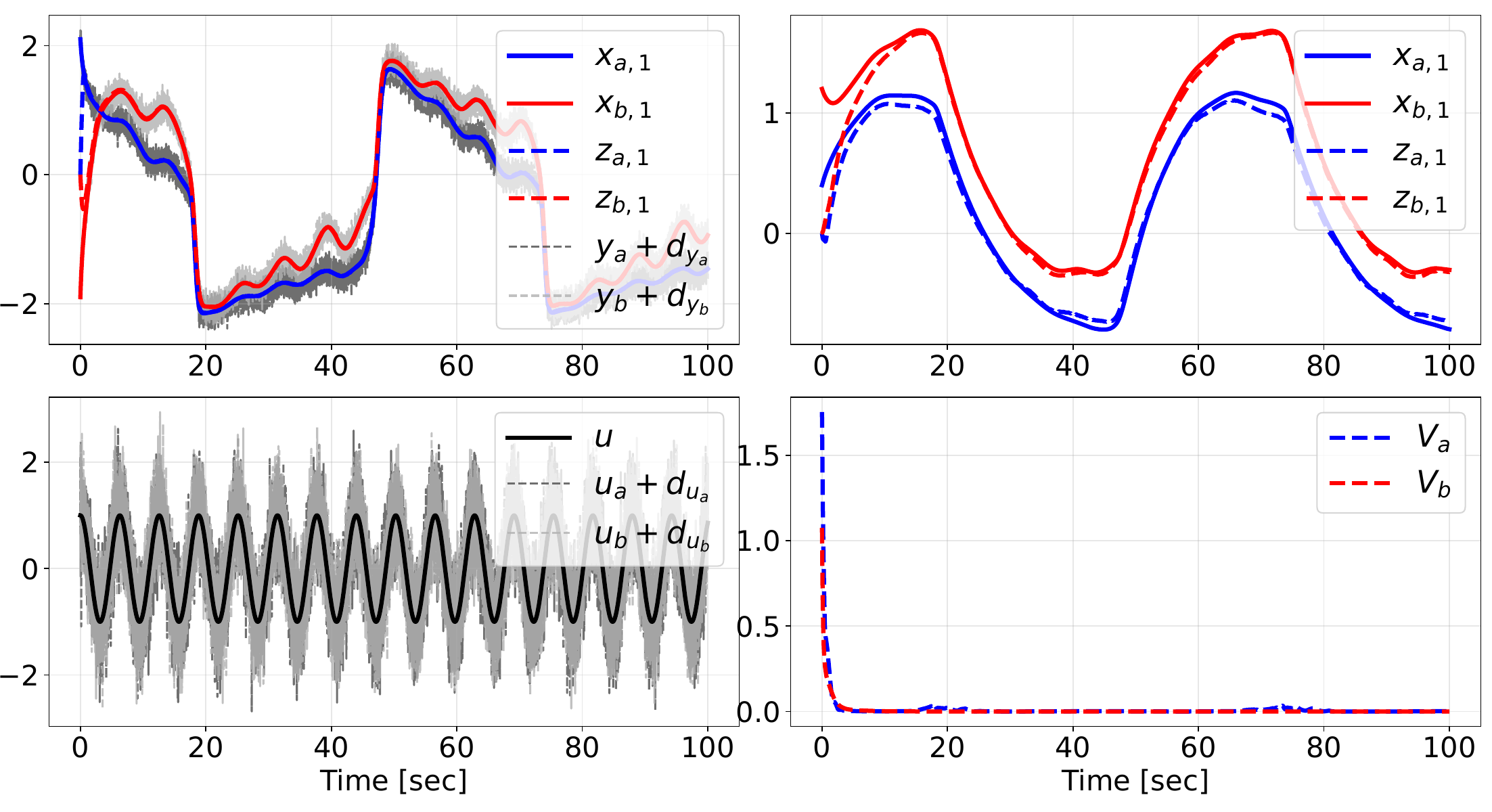}
    \caption{Distributed observer results on the interconnected FHN system with noisy measurements $y_a+d_{y_a}$ and $y_b+d_{y_b}$ ($d_{y_a},d_{y_b}\!\sim\!\mathcal{N}(0,0.1)$) and a broadcast input $u+d_u$ with common noise $d_u\!\sim\!\mathcal{N}(0,0.5)$.}
    \label{fig:fhn_network}
\end{figure}

Over $t\in[0,100]$ under noise--free setting, the RMSEs are $0.1056$ for subsystem $a$ and $0.1699$ for subsystem $b$. We then provide $y_a+d_{y_a}$, $y_b+d_{y_b}$, and the broadcast input $u+d_u$ with $d_{y_a}, d_{y_b}\sim\mathcal{N}(0,0.1)$ and $d_u\sim\mathcal{N}(0,0.5)$ to local observers. Figure~\ref{fig:fhn_network} shows the performance of the distributed observer on $[0, 100]$ interval, where RMSE recorded $0.1209$ for subsystem $a$ and $0.1846$ for subsystem $b$.
%
\section{Conclusion}
We proposed a framework for learning robust full--state observers equipped with a $\delta$ISS Lyapunov certificate for unknown dynamics. Practical convergence of the learned observer is proved in terms of learning errors. This renders the proposed framework \emph{safety--compatible} with certificate--based controllers for safety--critical applications. We further presented a distributed observer design for two interconnected subsystems via small--gain arguments. A natural direction for future research is extending the method to networks with an arbitrary number of nodes and topologies.
%
%


\section*{APPENDIX}

\subsection{Proof of Proposition~\ref{prop:disturbance_gain}} \label{appendix:a}
\begin{proof}
    From Grönwall's inequality, it follows 
    \begin{align} \label{eqA:gronwall}
        &V(t) \le V(0)\,e^{-\alpha t}
        + \frac{1 \!-\! e^{-\alpha t}}{\alpha}
        \Big(\rho_u\|u_1 \!-\! u_2\|_\infty^{2} \!+\! \rho_y\|y_1 \!-\! y_2\|_\infty^{2}\Big) \notag \\
        &\le V(0)\,e^{-\alpha t}
        + \frac{\rho_u}{\alpha}\|u_1 \!-\! u_2\|_\infty^{2}
        + \frac{\rho_y}{\alpha}\|y_1 \!-\! y_2\|_\infty^{2}.
    \end{align}
    Since $V(z_1, z_2) \geq \varepsilon |z_1 - z_2|^2$, it follows from~\eqref{eqA:gronwall} that 
    \begin{align}
        & |z_1 \!-\! z_2| \leq \sqrt{\frac{V(0)}{\varepsilon} e^{-\alpha t} \!+\! \frac{\rho_u}{\alpha\varepsilon}\left\|u_1 - u_2\right\|_{\infty}^2 \!+\! \frac{\rho_y}{\alpha\varepsilon}\left\|y_1 - y_2\right\|_{\infty}^2} \notag \\
        & \leq \sqrt{\frac{V(0)}{\varepsilon}} e^{-\frac{\alpha t}{2}} \!+\! \sqrt{\frac{\rho_u}{\alpha \varepsilon}} \left\|u_1 \!-\! u_2\right\|_{\infty} \!+\! \sqrt{\frac{\rho_y}{\alpha \varepsilon}} \left\|y_1 \!-\! y_2\right\|_{\infty},
    \end{align}
    where the property $\sqrt{a + b + c} \leq \sqrt{a} + \sqrt{b} + \sqrt{c}$ for nonnegative $a$, $b$, and $c$ is used.
\end{proof}



\bibliographystyle{IEEEtran}
\bibliography{Ref_Papers}

\end{document}